\documentclass[conference,a4paper]{IEEEtran}
\usepackage{amsfonts,amssymb,amstext,url}
\usepackage{amsthm}
\usepackage{color,dsfont,enumerate,hyperref,makeidx,multicol,verbatim,xfrac,xr}
\usepackage[cmex10]{amsmath}
\usepackage{pict2e}
\usepackage[mathscr]{euscript}
\usepackage{shadethm}
\usepackage[geometry]{ifsym}
\usepackage[utf8]{inputenc} 
\usepackage[T1]{fontenc}
\usepackage{url}
\usepackage{ifthen}
\usepackage{cite}
\usepackage[cmex10]{amsmath} 
\hypersetup{
	pdftitle		={Refined Strong Converse for the Constant Composition Codes},
	pdfauthor		={Baris Nakiboglu and Hao-Chung Cheng},
	pdfsubject		={},
	pdfkeywords		={}, 
}
\hypersetup{hidelinks}

\theoremstyle{plain}
\newtheorem{lemma}{Lemma} 
\newtheorem{theorem}{Theorem}

\newtheorem*{conjecture*}{Conjecture}

\theoremstyle{definition}
\newtheorem{definition}{Definition}

\newtheorem{remark}{Remark}
\newtheorem*{remark*}{Remark}
\definecolor{mygray}{gray}{0.4}
\newcommand{\set} [1]			{{\mathscr{{#1}}}}
\newcommand{\alg}[1]			{{\mathcal{{#1}}}}
\newcommand{\rndv}[1]      {{\mathsf{{#1}}}}

\newcommand{\msr}[1]       {{\it    {{#1}}}}
\newcommand{\cnst}[1]      {{\mathit{{#1}}}}

\newcommand{\integers}[1]	{{\mathbb{Z}}_{^{{#1}}}}

\newcommand{\reals}[1]		{{\mathbb{R}}_{^{{#1}}}}

\newcommand{\dif}[1]       {{\mathrm{d}{#1}}}  
\newcommand{\der}[2]        {\tfrac{\dif{#1}}{\dif{#2}}}  
\newcommand{\pder}[2]       {\tfrac{\partial{#1}}{\partial{#2}}}  
\newcommand{\supp}[1]       {\mathtt{supp}({{#1}})}       

\newcommand{\DEF}[0]			{{\!\!~\triangleq\!~}}  
\newcommand{\mtimes}[0]			{{\circledast}}

\newcommand{\AC}[0]            {{\prec}}

\newcommand{\abs}[1]           {{\left\lvert{{#1}}\right\lvert}}
\newcommand{\abp}[1]           {{\left\lvert{{#1}}\right\lvert^{+}}}
\newcommand{\lon}[1]           {{{\left\lVert{{#1}}\right\lVert}}} 

\newcommand{\IND}[1]           {{\mathds{1}_{\{#1\}}}}

\newcommand{\knd}[0]           {{\kappa}}
\newcommand{\tin}[0]           {{\cnst{t}}}
\newcommand{\blx}[0]           {{\cnst{n}}}

\newcommand{\EXS}[2]         {{\bf E}_{{#1}}\!\left[{#2}\right]}
\newcommand{\VXS}[2]         {{\bf V}_{{#1}}\!\left[{#2}\right]}

\newcommand{\fX}[0]          {{\cnst{f}}}

\newcommand{\RD}[3]				
{{\cnst{D}}_{{#1}}            \!\left(\left.            \! {#2}\right\Vert {#3}                  \right)}
\newcommand{\CRD}[4]			
{{\cnst{D}}_{{#1}}            \!\left(\left.\!\left.    \! {#2}\right\Vert {#3} \right\vert{{#4}}\right)}

\newcommand{\RMI}[3]	{{\cnst{I}}_{{#1}}        \!\left(        \! {#2};  \!{#3}                \!\right)}

\newcommand{\htdelta}[0] {{\cnst{\varDelta}}}
\newcommand{\composition}[0]	{{\cnst{\varUpsilon}}}

\newcommand{\rfm}[0]			{{{\msr{\nu}}}}

\newcommand{\dsce}[1]			{{\cnst{E}_{sc\!}'}       \left({#1}\right)}

\newcommand{\sce}[1]			{{\cnst{E}_{sc\!}}       \left({#1}\right)}
\newcommand{\rate}[0]			{{\cnst{R}}}
\newcommand{\GCD}[1]			{{{{\cnst{\Phi}}}\left({{#1}}\right)}}

\newcommand{\cln}[1]          {{{\xi}_{{#1}}}}

\newcommand{\rnb}[0]          {{\cnst{\beta}}}

\newcommand{\rno}[0]          {{\cnst{\alpha}}}

\newcommand{\rns}[0]          {\rno^{\!\ast}}

\newcommand{\Pem}[1]           {{\it P_{{{\bf e}}}^{{#1}}}}         
 
\newcommand{\enc}[0]           {{\varPsi}} 
\newcommand{\dec}[0]           {{\varTheta}}    

\newcommand{\oev}[0]           {{\set{E}}}

\newcommand{\pmea}[1]          {{{\alg{P}}({#1})}}

\newcommand{\pdis}[1]          {{{\set{P}}({#1})}}

\newcommand{\dinp}[0]          {{\cnst{x}}}

\newcommand{\inpS}[0]          {{\set{X}}}

\newcommand{\dout}[0]          {{\cnst{y}}}
\newcommand{\out}[0]           {{\rndv{Y}}}
\newcommand{\outS}[0]          {{\set{Y}}}
\newcommand{\outA}[0]          {{\alg{Y}}}

\newcommand{\dsta}[0]          {{\cnst{z}}}

\newcommand{\dmes}[0]          {{\cnst{m}}}

\newcommand{\mesS}[0]          {{\set{M}}}

\newcommand{\estS}[0]          {{\widehat{{\set{M}}}}}

\newcommand{\mean}[0]        {{{\msr{\mu}}}}

\newcommand{\mA}[0]				{{\msr{a}}}    
\newcommand{\amn}[1]			{{{\mA}_{{#1}}}}

\newcommand{\mP}[0]				{{\msr{p}}}

\newcommand{\mQ}[0]				{{\msr{q}}}    
\newcommand{\qmn}[1]			{{{\mQ}_{{#1}}}}

\newcommand{\mS}[0]				{{\msr{s}}}

\newcommand{\mV}[0]				{{\msr{v}}}    
\newcommand{\vmn}[1]			{{{\mV}_{{#1}}}}

\newcommand{\Vm}[0]				{{{\cnst{V}}}}

\newcommand{\mW}[0]				{{\msr{w}}}    
\newcommand{\wmn}[1]			{{{\mW}_{{#1}}}}
\newcommand{\wma}[2]			{{{\mW}_{{#1}}^{{#2}}}}
\newcommand{\Wm}[0]				{{{\cnst{W}}}}
\newcommand{\Wmn}[1]			{{{\cnst{W}}_{{#1}}}}
\newcommand{\Wma}[2]			{{{\cnst{W}}_{{#1}}^{{#2}}}}

\newcommand{\csiszar}[0]							{Csisz\'{a}r~}

\newcommand{\korner}[0]								{K\"{o}rner~}
\newcommand{\renyi}[0]								{R\'{e}nyi~}

\makeatletter
\DeclareRobustCommand{\bigplus}{%
	\mathop{\vphantom{\sum}\mathpalette\@bigplus\relax}\slimits@
}
\newcommand{\@bigplus}[2]{\vcenter{\hbox{\make@bigplus{#1}}}}
\newcommand{\make@bigplus}[1]{%
	\sbox\z@{$\m@th#1\sum$}%
	\setlength{\unitlength}{\wd\z@}%
	\begin{picture}(1.4,1.4)
	\linethickness{.17ex}
	\Line(.7,.14)(.7,1.26)
	\Line(.14,.7)(1.26,.7)
	\end{picture}%
}
\DeclareRobustCommand{\bigtimes}{%
	\mathop{\vphantom{\sum}\mathpalette\@bigtimes\relax}\slimits@
}
\newcommand{\@bigtimes}[2]{\vcenter{\hbox{\make@bigtimes{#1}}}}
\newcommand{\make@bigtimes}[1]{%
	\sbox\z@{$\m@th#1\sum$}%
	\setlength{\unitlength}{\wd\z@}%
	\begin{picture}(1,1)
	\linethickness{.17ex}
	\Line(.1,.1)(.9,.9)
	\Line(.1,.9)(.9,.1)
	\end{picture}%
}
\makeatother
\IEEEoverridecommandlockouts 
\begin{document}
	\pagestyle{plain}	
	\pagenumbering{arabic}
\title{\vspace{-.4em}Refined Strong Converse for\\ the Constant Composition Codes\vspace{-.4em}}
\author{%
  \IEEEauthorblockN{Hao-Chung Cheng}
   \IEEEauthorblockA{Department of Applied Mathematics and Theoretical Physics\\
                    University of Cambridge\\
                    Cambridge   CB3  0WA,  U.K.\\
                    Email: haochung.ch@gmail.com\vspace{-2em}}
  \and
  \IEEEauthorblockN{Bar\i\c{s} Nakibo\u{g}lu}
   \IEEEauthorblockA{Department of Electrical and Electronics Engineering\\
                    Middle East Technical University (METU)\\ 
                    06800 Ankara, Turkey\\
                    Email:bnakib@metu.edu.tr\vspace{-2em}}
  \thanks{The work was done partially while the authors were visiting 
  	the Institute for Mathematical Sciences, National University of Singapore in 2017. 
  	The visit was supported by the Institute. This work is also supported by
  	the Ministry of Science and Technology, Taiwan (R.O.C.), 
  	under Grant 108-2917-I-564-042 and 109WXA0310019,
  	the Science Academy, Turkey, under The Science Academy's Young Scientist Awards 
  	Program (BAGEP),
  	and by the Scientific and Technological Research Council 
	of Turkey (T\"{U}B\.{I}TAK) under Grant 119E053.}
}
\maketitle 

\begin{abstract}
A strong converse bound for constant composition codes of the form 
\(\Pem{(\blx)}\geq 1- A \blx^{-0.5(1-\dsce{\rate,\Wm\!,\mP})} e^{-\blx\sce{\rate,\Wm\!,\mP}}\)
is established using the Berry-Esseen theorem through the concepts of 
Augustin information and Augustin mean,
where \(A\) is a constant determined by the channel \(\Wm\), 
the composition \(\mP\), and the rate \(\rate\), 
i.e., \(A\) does not depend on the block length \(\blx\). 
\end{abstract}

\section{Introduction}
On discrete stationary product channels, the error probability of 
codes operating at rates above capacity is not only bounded away 
from zero but also converging to one. This property, first observed 
by Wolfowitz \cite{wolfowitz57}, is called the strong converse property. 
For arbitrary stationary product channels, a necessary and sufficient condition 
for the strong converse property was determined by Augustin 
\cite[\S 10]{augustin66}, \cite[\S 13]{augustin78}. 
The strong converse property does not hold in general; nevertheless
there does exist a universal asymptotic constant that bounds the error probability 
of codes operating at rates above the capacity of the stationary product 
channels, according to Beck and Csiszar \cite{beckC78}.
In \cite{verduH94}, Verd\'{u} and Han 
provided a necessary and sufficient condition 
for the strong converse property for channels that are not necessarily 
stationary or memoryless.

Using the concept of \renyi capacity, which was employed earlier by Gallager \cite{gallager65}
for analyzing the error probability of codes operating at rates below the channel capacity, 
Arimoto established in \cite{arimoto73} the following lower bound to the error probability 
of codes on discrete stationary product channels (DPSCs) operating at a rate \(\rate\) above the channel capacity:
\begin{align}
\label{eq:arimoto}
\Pem{(\blx)}\geq 1- e^{-\blx\sce{\rate}}
\end{align}
where \(\sce{\cdot}\) is the strong converse exponent of the channel.
Although Arimoto's initial proof in \cite{arimoto73} is for DSPCs,
Arimoto's lower bound can be proved as a one shot bound 
for more general channel models using 
Jensen's inequality or H\"olders inequality
as noted by Augustin \cite{augustin78} and Sheverdyaev \cite{sheverdyaev82},
see also \cite{nakiboglu19B,polyanskiyV10,nagaoka01}.
Arimoto's lower bound is used to establish the strong converse on 
channels for which alternative derivations of the strong converse 
is much more tedious, e.g. quantum channels discussed in
\cite{nagaoka01,ogawaN99,konigW09,mosonyiH11,sharmaW13,wildeWY14,guptaW15,mosonyiO17,tomamichelWW17,CHDH-2018, CHDH2-2018}
and Poisson channels mentioned in \cite[Appendices B-B and B-C]{nakiboglu19B}.

Aritmoto's lower bound to the error probability has been derived for 
certain constrained codes on memoryless channels as well, 
see 
Dueck and \korner \cite{dueckK79} for the constant composition codes on DSPCs,
Oohama \cite{oohama17A} for the Gaussian channel,
and Cheng \emph{et al.} \cite{CHDH2-2018} and 
Mosonyi and Ogawa \cite{mosonyiO18} for the 
constant composition codes on classical-quantum channels.

For codes on DSPCs Omura \cite{omura75} has shown
\begin{align}
\label{eq:DK}
\lim\nolimits_{\blx \to \infty} -\tfrac{1}{\blx} \ln(1-\Pem{(\blx)})
&\!\leq\! \sce{\rate}
\end{align}
for all rates above the channel capacity and below a certain threshold.
Dueck and \korner \cite{dueckK79} established \eqref{eq:DK} for all 
rates above the channel capacity.
Thus Arimoto's bound is tight, in terms of the exponential decay rate of 
the probability of correct decoding with block length, 
for all rates above the channel capacity.
An analogous result is derived  
for constant composition codes on DSPCs in \cite{dueckK79},
for the Gaussian channel in \cite{oohama17A},
for classical-quantum channels in \cite{mosonyiO17},
for classical data compression with quantum side information in \cite{CHDH-2018},
and for constant composition codes on classical-quantum channels in \cite{mosonyiO18}.

Although Arimoto's bound, given in \eqref{eq:arimoto},
is tight in terms of the exponential 
decay rate of the correct decoding probability with block length,
the prefactor multiplying the exponentially decaying term can be improved.
In particular,
for the constant composition codes operating at rates larger than 
the mutual information of the composition.
Theorems \ref{thm:constantcomposition-RSC-lowrate}
and \ref{thm:constantcomposition-RSC-highrate}, 
in the following, establish a strong converse bound of the form
\begin{align}
\label{eq:RSC}
\Pem{(\blx)}
&\geq 1- A n^{-0.5(1-\dsce{\rate})} e^{-\blx\sce{\rate}}
\end{align}
where \(\sce{\cdot}\) is the strong converse exponent 
and \(\dsce{\cdot}\) is its derivative with respect 
to the rate.
Since \(1\geq \dsce{\rate}\geq 0\) for all rates \(\rate\) 
and \(\dsce{\rate}<1\) for small enough rates \(\rate\),
the bound \eqref{eq:RSC} improves \eqref{eq:arimoto} strictly.
In accordance with the corresponding improvements 
of the sphere packing bound for rates below the channel capacity
given in \cite{altugW11,altugW14A,chengHT19,nakiboglu19-ISIT,nakiboglu20F},
we call the bounds of the form given in \eqref{eq:RSC}
refined strong converses. 

Proof of Theorem \ref{thm:constantcomposition-RSC-lowrate} is analogous to 
the proof of refined sphere packing bound presented in 
\cite{nakiboglu19-ISIT,nakiboglu20F}:
it relies on a  tight characterization of the 
trade-off between type-I and type-II error probabilities in 
the hypothesis testing problem  
with (possibly non-stationary) independent samples
through the concepts of Augustin information and mean.
However, in \cite{nakiboglu19-ISIT,nakiboglu20F} 
for the regime of interest 
the optimal tilting parameter is between zero and one;
whereas we are now interested in the regime where 
the optimal tilting parameter is larger than one.
Similarly, in \cite{nakiboglu19-ISIT,nakiboglu20F}
Agustin information measures for orders between zero and one 
are used together with the sphere packing exponent,
whereas we employ
Agustin information measures for orders larger than one 
together with the strong converse exponent in our analysis. 

We conclude this section with an overview of the paper.
In \S\ref{sec:model-and-notation}, we describe our model and notation.
In \S\ref{sec:HypothesisTesting}, we 
employ the concept of tilted probability measure and
the Berry-Esseen theorem to obtain a lower bound 
on the type-II error probability in hypothesis testing problem with independent 
---but not necessarily identically distributed---samples
for the regime where the optimal tilting parameter is larger than one.
In \S\ref{sec:preliminary}, we review Augustin's information measures 
and the strong converse exponent.
In \S\ref{sec:RSC}, we establish a refined strong converse 
for the constant composition codes on stationary memoryless channels.  
We conclude our presentation with a brief discussion
of the results and future work in \S\ref{sec:conclusion}.

\section{Model and Notation}\label{sec:model-and-notation}
We denote 
the set of all probability mass functions that 
are positive only for finitely many elements of \(\inpS\) by \(\pdis{\inpS}\)
and the set of all probability measures on a measurable space \((\outS,\outA)\)
by \(\pmea{\outA}\).
The \({\cal L}^{1}\) norm of a measure \(\mean\) is denoted by \(\lon{\mean}\).
The expected value and variance of a measurable function \(\fX\) 
under the probability measure \(\mean\) are denoted 
by
\(\EXS{\mean}{\fX}\) and \(\VXS{\mean}{\fX}\).
The Cartesian product of sets \(\inpS_{1},\ldots,\inpS_{\blx}\) 
is denoted by \(\inpS_{1}^{\blx}\);
the product of \(\sigma\)-algebras \(\outA_{1},\ldots,\outA_{\blx}\)
is denoted by \(\outA_{1}^{\blx}\).
The symbol \(\otimes\) is used to denote both products of \(\sigma\)-algebras
and products of measures.
 
A \emph{channel} \(\Wm\) is a function from \emph{the input set} \(\inpS\) to  the set of all probability 
measures on \emph{the output space} \((\outS,\outA)\):
\begin{align}
\label{eq:def:channel}
\Wm:\inpS \to \pmea{\outA}.
\end{align}
If \(\inpS\) and \(\outA\) are both finite sets, then \(\Wm\) is a \emph{discrete channel}.
The product of \(\Wmn{\tin}\!:\!\inpS_{\tin}\!\to\!\pmea{\outA_{\tin}}\) 
for \(\tin\!\in\!\{1,\ldots,\blx\}\) 
is a channel of the form \(\Wmn{[1,\blx]}:\inpS_{1}^{\blx}\to\pmea{\outA_{1}^{\blx}}\) 
satisfying 
\begin{align}
\label{eq:def:product}
\Wmn{[1,\blx]}(\dinp_{1}^{\blx})
&=\bigotimes\nolimits_{\tin=1}^{\blx}\Wmn{\tin}(\dinp_{\tin})
&
&\forall \dinp_{1}^{\blx}\in\inpS_{1}^{\blx}.
\end{align}
Any channel obtained by curtailing the input set of a length \(\blx\)
product channel is called a length \(\blx\) \emph{memoryless channel}.
A product channel \(\Wmn{[1,\blx]}\) is \emph{stationary} iff \(\Wmn{\tin}=\Wm\)
for all \(\tin\)'s for some \(\Wm\!\).
On a stationary channel, we denote the composition 
(i.e. the empirical distribution, the type) 
of each \(\dinp_{1}^{\blx}\) by \(\composition(\dinp_{1}^{\blx})\);
thus \(\composition(\dinp_{1}^{\blx})\in\pdis{\inpS}\).

An \((M,L)\) \emph{channel code} on \(\Wmn{[1,\blx]}\)
is composed of an 
\emph{encoding function} \(\enc\) 
from the message set \(\mesS\DEF\{1,2,\ldots,M\}\) to the input set \(\inpS_{1}^{\blx}\)
and a \(\outA_{1}^{\blx}\)-measurable \emph{decoding function} \(\dec\)
from the output set \(\outS_{1}^{\blx}\) to \(\estS\DEF\{\set{L}:\set{L}\subset\mesS \mbox{~and~}\abs{\set{L}}\leq L\}\).
For any channel code \((\enc,\dec)\) on \(\Wmn{[1,\blx]}\),
\emph{the conditional error probability} \(\Pem{\dmes}\) for \(\dmes \in \mesS\) 
and \emph{the average error probability} \(\Pem{}\)
are defined as
\begin{align}
\notag
\Pem{\dmes}
&\DEF \EXS{\Wmn{[1,\blx]}(\enc(\dmes))}{\IND{\dmes\notin\dec(\out_{1}^{\blx})}},
\\
\notag
\Pem{} 
&\DEF\tfrac{1}{M} \sum\nolimits_{\dmes\in \mesS} \Pem{\dmes}.
\end{align}
A channel code is a constant composition code iff all of its codewords  have 
the same composition, i.e. \(\exists\mP\in\pdis{\inpS}\) such that
\(\composition(\enc(\dmes))=\mP\),  \(\forall\dmes\in\mesS\).

\section{Hypothesis Testing Problem, Tilted Probability Measure, and Berry Esseen Theorem}\label{sec:HypothesisTesting}
Our main aim in this section is to characterize the trade-off between
type-I and type-II error probabilities
using the Berry-Essen theorem and the tilted probability measure.
This trade-off can be studied in various regimes;
in order to specify the regime of interest, let us first 
recall the definition of \renyi divergence and
define the tilted probability measure.
\begin{definition}\label{def:divergence}
	For any \(\rno\in\reals{+}\) and \(\mW,\mQ\in\pmea{\outA}\),
	\emph{the order \(\rno\) \renyi divergence between \(\mW\) and \(\mQ\)} is
	\begin{align}
\notag
	\RD{\rno}{\mW}{\mQ}
	&\DEF \begin{cases}
	\tfrac{1}{\rno-1}\ln \int (\der{\mW}{\rfm})^{\rno} (\der{\mQ}{\rfm})^{1-\rno} \rfm(\dif{\dout})
	&\rno\neq 1\\
	\int  \der{\mW}{\rfm}\left[ \ln\!\der{\mW}{\rfm} -\ln\!\der{\mQ}{\rfm}\right] \rfm(\dif{\dout})
	&\rno=1
	\end{cases}
	\end{align}
	where \(\rfm\) is any measure satisfying \(\mW\AC\rfm\) and \(\mQ\AC\rfm\).
\end{definition}

\begin{definition}\label{def:tiltedprobabilitymeasure}
	For any \(\mW,\mQ\in\pmea{\outA}\), 
	let \(\wmn{ac}\) be the component of \(\mW\) that is absolutely
	continuous in \(\mQ\).
	If \(\lon{\wmn{ac}}\neq0\), then 
	\emph{the order \(1\) tilted probability measure} \(\wma{1}{\mQ}\) is
	\begin{align}
	\notag
	\wma{1}{\mQ}
	&\DEF\tfrac{\wmn{ac}}{\lon{\wmn{ac}}}.
	\end{align}
	Furthermore, for any \(\rno\in\reals{+}\) satisfying 
	\(\RD{\rno}{\wma{1}{\mQ}}{\mQ}<\infty\), 
	\emph{the order \(\rno\) tilted probability measure} \(\wma{\rno}{\mQ}\) 
	is defined in terms of its Radon-Nikodym derivative with respect to 
	\(\mQ\) as follows
\begin{align}
\notag 
\der{}{\mQ}\wma{\rno}{\mQ}
&\DEF e^{(1-\rno)\RD{\rno}{\wma{1}{\mQ}}{\mQ}}\left(\der{\wma{1}{\mQ}}{\mQ}\right)^{\rno}.
\end{align}
\end{definition}

The definition of tilted probability measure used in \cite{nakiboglu19B,nakiboglu19C,nakiboglu20F}, 
employs \(\mW\) in the place of \(\wma{1}{\mQ}\).
Whenever \(\wma{1}{\mQ}=\mW\), i.e. whenever \(\mW\AC\mQ\), 
it is equivalent to Definition \ref{def:tiltedprobabilitymeasure}.
For orders in \((0,1)\) these two definitions are equivalent 
even if \(\wma{1}{\mQ}\neq\mW\).
For orders larger than or equal to one, they differ only when
\(\wma{1}{\mQ}\neq\mW\) and \(\RD{\rno}{\wma{1}{\mQ}}{\mQ}<\infty\).
In this case, \(\RD{\rno}{\mW}{\mQ}=\infty\) for all \(\rno\)'s 
in \([1,\infty)\) and 
\(\wma{\rno}{\mQ}\) is not defined according to the definition used  
\cite{nakiboglu19B,nakiboglu19C,nakiboglu20F}
but \(\wma{\rno}{\mQ}\) is defined according to 
Definition \ref{def:tiltedprobabilitymeasure}.

In order to see why Definition \ref{def:tiltedprobabilitymeasure}
can be more relevant than the one used in
\cite{nakiboglu19B,nakiboglu19C,nakiboglu20F}, 
let us consider two probability measures
\(\mW\) and \(\mQ\) for which 
\(\wma{1}{\mQ}\neq\mW\) and \(\RD{\rnb}{\wma{1}{\mQ}}{\mQ}\) is finite
for some \(\rnb>1\).
Then both 
\(\RD{1}{\wma{\rno}{\mQ}}{\mQ}\) and 
\(\RD{1}{\wma{\rno}{\mQ}}{\wma{1}{\mQ}}\)
are analytic, and hence continuous, functions of the order
 \(\rno\) on \((0,\rnb)\)
by \cite[Lemma \ref*{C-lem:analyticity}]{nakiboglu19C}.
On the other hand as, a result of Pinsker's inequality 
\cite[Thm. 31]{ervenH14}
we have
\begin{align}
\notag
\lon{\wma{\rno}{\mQ}-\wma{1}{\mQ}}
&\leq
\sqrt{2 \RD{1}{\wma{\rno}{\mQ}}{\wma{1}{\mQ}}}.
\end{align}
Thus \(\wma{\rno}{\mQ}\) converges in total variation to
\(\wma{1}{\mQ}\), rather than \(\mW\), 
as \(\rno\) converges to one by the continuity 
of \(\RD{1}{\wma{\rno}{\mQ}}{\wma{1}{\mQ}}\) in \(\rno\).
Furthermore, 
the continuity of \(\RD{1}{\wma{\rno}{\mQ}}{\mQ}\) in 
\(\rno\) implies that
\begin{align}
\notag
\lim\nolimits_{\rno\uparrow 1} \RD{1}{\wma{\rno}{\mQ}}{\mQ}
&=\RD{1}{\wma{1}{\mQ}}{\mQ}.
\end{align}
This convergence provides further justification 
to Definition \ref{def:tiltedprobabilitymeasure}
because \(\RD{1}{\mW}{\mQ}=\infty\).
Recall that both definitions of 
the tilted probability measure lead to the same 
\(\wma{\rno}{\mQ}\) for \(\rno\in(0,1)\).

\begin{lemma}\label{lem:SC-HTBE-converse}
For any \(\rno\!\in\!(1,\infty)\),
\(\blx\!\in\!\integers{+}\), 
\(\wmn{\tin},\qmn{\tin}\!\in\!\pmea{\outA_{\tin}}\), 
let \(\wmn{\tin,ac}\) be the component of \(\wmn{\tin}\) that is absolutely continuous in \(\qmn{\tin}\)
and let \(\amn{2}\), \(\amn{3}\), and \(\htdelta\) be
\begin{align}
\notag
\amn{2}
&\!\DEF\!\tfrac{1}{\blx}\!\sum\nolimits_{\tin=1}^{\blx} 
\EXS{\wma{\rno}{\mQ}}{\left(\ln\!\der{\wmn{\tin,ac}}{\qmn{\tin}}-\EXS{\wma{\rno}{\mQ}}{\ln\!\der{\wmn{\tin,ac}}{\qmn{\tin}}}\right)^{2}},
\\
\notag
\amn{3}
&\!\DEF\!\tfrac{1}{\blx}\!\sum\nolimits_{\tin=1}^{\blx} 
\EXS{\wma{\rno}{\mQ}}{\abs{\ln\!\der{\wmn{\tin,ac}}{\qmn{\tin}}-\EXS{\wma{\rno}{\mQ}}{\ln\!\der{\wmn{\tin,ac}}{\qmn{\tin}}}}^{3}},
\\
\notag
\htdelta
&\!\DEF\! \tfrac{1}{e\sqrt{\amn{2}}}\left(\tfrac{1}{\sqrt{2\pi}}+2\tfrac{0.56\amn{3}}{\amn{2}}\right),
\end{align}
where \(\mW\!=\!\otimes_{\tin=1}^{\blx} \wmn{\tin}\) and \(\mQ\!=\!\otimes_{\tin=1}^{\blx} \qmn{\tin}\).
Then for any \(\oev\in\outA_{1}^{\blx}\) 
and \(\rnb\!\in\!\reals{+}\)
satisfying \(\mQ(\oev)\!\leq\!\rnb e^{-\RD{1}{\wma{\rno}{\mQ}}{\mQ}}\),
we have
\begin{align}
\label{eq:lem:SC-HTBE-converse}
\hspace{-.35cm}\mW({\outS_{1}^{\blx}}\!\setminus\!\oev\!)
&\!\geq\!
\left[\prod\limits_{\tin=1}^{\blx}\!\lon{\wmn{\tin,ac}}\!\right]
-\!\tfrac{2e^{\rno} \htdelta^{\frac{1}{\rno}} \rnb^{\frac{\rno-1}{\rno}}}{(\rno-1)^{\sfrac{1}{\rno}}}
\blx^{-\frac{1}{2\rno}}e^{-\RD{1}{\wma{\rno}{\mQ}}{\mW}},\!
\\
\label{eq:lem:SC-HTBE-converse-alternative}
&\!=\!
\left[1\!-\!
\tfrac{2e^{\rno} \htdelta^{\frac{1}{\rno}} \rnb^{\frac{\rno-1}{\rno}}}{(\rno-1)^{\sfrac{1}{\rno}}}
\tfrac{e^{-\RD{1}{\wma{\rno}{\mQ}}{\wma{1}{\mQ}}}}{\blx^{\sfrac{1}{2\rno}}}
\right]
\prod\limits_{\tin=1}^{\blx}\lon{\wmn{\tin,ac}}.
\end{align}
\end{lemma}

Lemma \ref{lem:SC-HTBE-converse} is often applied for the case when \(\wmn{\tin}\) 
is absolutely continuous in \(\qmn{\tin}\) for all \(\tin\),
i.e. in the case when \(\wmn{\tin}\AC\qmn{\tin}\) for all \(\tin\).
In that case \(\wmn{\tin,ac}=\wmn{\tin}\) for all  \(\tin\)
and thus \(\prod\nolimits_{\tin=1}^{\blx}\lon{\wmn{\tin,ac}}=1\).	

The lower bound asserted in Lemma \ref{lem:SC-HTBE-converse} is tight
in the sense that for any \(\rno\!\in\!(1,\infty)\) and 
\(\rnb\!\in\!\left[\tfrac{9\htdelta e^{\rno \delta}}{\sqrt{\blx}}
e^{-\rno\sqrt{\amn{2}\blx}},
\tfrac{9\htdelta}{\sqrt{\blx}} e^{\rno\sqrt{\amn{2}\blx}}\right]\),
there exists an \(\oev\!\in\!\outA_{1}^{\blx}\) such that
\begin{align}
\notag
	\mQ(\oev)
	&\!\leq\!\rnb e^{-\RD{1}{\wma{\rno}{\mQ}}{\mQ}},	
	\\
	\notag
	\hspace{-.2cm}
	\mW({\outS_{1}^{\blx}}\setminus\oev)
	&\!\leq\! \left[\prod\limits_{\tin=1}^{\blx}\lon{\wmn{\tin,ac}}\right]
	-\tfrac{e^{(1-\rno)\delta}}{\sqrt{2\pi\amn{2}}}
	(\tfrac{\rnb}{9\htdelta})^{\frac{\rno-1}{\rno}}
	\tfrac{e^{-\RD{1}{\wma{\rno}{\mQ}}{\mW}}}{\blx^{\sfrac{1}{2\rno}}},
	\end{align}
where \(\delta=e\sqrt{2\pi e\amn{2}}\htdelta\),
see Appendix \ref{sec:SC-HTBE-achievability} for a proof. 

One can calculate exact asymptotic value of the constant
in the trade-off between error probabilities
in the hypothesis testing problem in this regime,
under stronger hypotheses. 
For the stationary case
---i.e. the case when \(\wmn{\tin}\!=\!\wmn{1}\), \(\qmn{\tin}\!=\!\qmn{1}\) for all \(\tin\)--- 
assuming \(\mW\AC\mQ\),
first 
\csiszar and\footnote{The approach of \cite{csiszarL71} is sound, but its calculations
	seem to have some mistakes.} Longo \cite{csiszarL71} 
and more recently Vazquez-Vilar \emph{et al.} \cite{vazquezFKL18}
have discussed this problem.

\begin{proof}[Proof of Lemma \ref{lem:SC-HTBE-converse}]
	Let the random variables \(\cln{\tin}\) and \(\cln{}\) 
	be\footnote{\(\cln{\tin}\) and \(\cln{}\) are implicitly assumed 
		to be zero outside the support of \(\mQ\).}
	\begin{align}
	\notag
	\cln{\tin}
	&\!\DEF\!\ln\!\der{\wmn{\tin,ac}}{\qmn{\tin}},
	\\
	\notag
	\cln{}
	&\!\DEF\! \sum\nolimits_{\tin=1}^{\blx} \cln{\tin}.
	\end{align}
	Then \(\cln{}\!=\!\ln\!\der{\wmn{ac}}{\mQ}\), and hence 
	\(\cln{}\!=\!\ln\!\der{\mW}{\mQ}\),
	holds \(\mQ\)-a.s., 
	and the Radon-Nikodym derivatives 
	\(\der{\wma{\rno}{\mQ}}{\mQ}\) and \(\der{\wma{\rno}{\mQ}}{\mW}\)
	can be expressed in terms of \(\cln{}\)
	as follows
\begin{align}
\label{eq:CLQ}
\ln\der{\wma{\rno}{\mQ}}{\mQ}
&\!=\!\RD{1}{\wma{\rno}{\mQ}}{\mQ}+\rno\left(\cln{}-\EXS{\wma{\rno}{\mQ}}{\cln{}}\right), 
\\
\label{eq:CLW}
\ln\der{\wma{\rno}{\mQ}}{\mW}
&\!=\!\RD{1}{\wma{\rno}{\mQ}}{\!\mW}+(\rno\!-\!1)\left(\cln{}-\EXS{\wma{\rno}{\mQ}}{\cln{}}\right). 
\end{align}	
For each integer \(\knd\), let the set \(\set{B}_{\knd}\) be
\begin{align}
\label{eq:SC-HTBE-Bdef}
\set{B}_{\knd}
&\!\DEF\! \left\{\dout_{1}^{\blx}\!:\!\tau+\knd\leq \cln{}-\EXS{\wma{\rno}{\mQ}}{\cln{}}<\tau+\knd+1 \right\}. 
\end{align}
Then for any \(\oev\!\in\!\outA_{1}^{\blx}\) and \(\knd\!\in\!\integers{}\), 
we can bound \(\wma{\rno}{\mQ}(\oev\!\cap\!\set{B}_{\knd})\) from above
in terms of  \(\mQ(\oev \cap \set{B}_{\knd})\) using \eqref{eq:CLQ}
and from below 
in terms of  \(\mW(\oev \cap \set{B}_{\knd})\) using \eqref{eq:CLW},
as follows
\begin{align}
\label{eq:SC-HTBE-CB1}
\wma{\rno}{\mQ}(\oev \cap \set{B}_{\knd})
&\!\leq\!\mQ(\oev \cap \set{B}_{\knd}) e^{\RD{1}{\wma{\rno}{\mQ}}{\mQ}+\rno\tau+\rno(\knd+1)},
\\
\label{eq:SC-HTBE-CB2}
\wma{\rno}{\mQ}(\oev \cap \set{B}_{\knd})
&\!\geq\!\mW(\oev \cap \set{B}_{\knd})
e^{\RD{1}{\wma{\rno}{\mQ}}{\mW}+(\rno-1)\tau+(\rno-1)\knd}.
\end{align}
In order bound \(\mW({\outS_{1}^{\blx}}\setminus\oev)\) we use
\(\mW({\outS_{1}^{\blx}}\setminus\oev)\geq \mW(\set{B}_{\integers{}}\setminus\oev)\)
and
\(\mW(\set{B}_{\integers{}}\setminus\oev)=\mW(\set{B}_{\integers{}})-\mW(\set{B}_{\integers{}}\cap\oev)\)
where \(\set{B}_{\integers{}}\!\DEF\!\!\cup_{\knd\in\integers{}}\set{B}_{\knd}\). 
First note that  for any family of reference measures  \(\{\rfm_{\tin}\}\) satisfying 
\(\wmn{\tin}\AC\rfm_{\tin}\) and \(\qmn{\tin}\AC\rfm_{\tin}\) for all \(\tin\) we have
\begin{align}
\notag
\mW(\set{B}_{\integers{}})
&=\left(\bigotimes\nolimits_{\tau=1}^{\blx}\wmn{\tau}\right)
\left(\left\{\dout_{1}^{\blx}\!:\!\der{\wmn{\tin}}{\rfm_{\tin}}>0 \mbox{~and~} \der{\qmn{\tin}}{\rfm_{\tin}}>0 
\quad \forall \tin \right\}\right)
\\
\notag
&=\prod\nolimits_{\tin=1}^{\blx}\wmn{\tin} \left(
\left\{\dout_{\tin}\!:\!\der{\wmn{\tin}}{\rfm_{\tin}}>0 \mbox{~and~} \der{\qmn{\tin}}{\rfm_{\tin}}>0 \right\}
\right)
\\
\label{eq:SC-HTBE-BZ}
&=\prod\nolimits_{\tin=1}^{\blx}\lon{\wmn{\tin,ac}}.
\end{align}
Thus for 
\(\set{B}_{\integers{\leq0}\!}\!\!\DEF\!\cup_{\knd\in\integers{\leq0}}\!\set{B}_{\knd\!}\)
and
\(\set{B}_{\integers{+}}\!\DEF\!\cup_{\knd\in\integers{+}}\set{B}_{\knd}\),
we have
\begin{align}
\label{eq:SC-HTBE-converse-1}
\mW({\outS_{1}^{\blx}}\!\setminus\!\oev)
&\!\geq\!\prod\nolimits_{\tin=1}^{\blx}\!\lon{\wmn{\tin,ac}}
\!-\!\mW(\oev\!\cap\!\set{B}_{\integers{\leq0\!}})
\!-\!\mW(\oev\!\cap\!\set{B}_{\integers{+}\!}).
\end{align}
In order to bound \(\mW(\oev\!\cap\!\set{B}_{\integers{\leq0}})\)
we use 
\eqref{eq:SC-HTBE-CB1}, \eqref{eq:SC-HTBE-CB2}, 
the identity \(\mQ(\oev \cap \set{B}_{\knd})\leq \mQ(\oev)\),
the hypothesis  \(\mQ(\oev)\!\leq\!\rnb e^{-\RD{1}{\wma{\rno}{\mQ}}{\mQ}}\),
and the formula for the sum of a geometric series:
	\begin{align}
	\notag
	\mW(\oev \cap \set{B}_{\integers{\leq0}})
	&=\sum\nolimits_{\knd\in\integers{\leq 0}} \mW(\oev \cap \set{B}_{\knd})
	\\
	\notag
	&\leq\sum\nolimits_{\knd\in\integers{\leq 0}} 
	\rnb e^{\tau+\knd+\rno-\RD{1}{\wma{\rno}{\mQ}}{\mW}}
	\\
	\label{eq:SC-HTBE-converse-2}
	&\leq \rnb  e^{\tau+\rno-\RD{1}{\wma{\rno}{\mQ}}{\mW}}\tfrac{1}{1-e^{-1}}.
	\end{align}
	On the other hand \(\cln{\tin}\)'s are jointly independent under \(\wma{\rno}{\mQ}\).
	Thus the Berry-Esseen theorem \cite{berry41,esseen42,shevtsova10} implies
	\begin{align}
	\notag
	\wma{\rno}{\mQ}(\set{B}_{\knd})
	&\leq \GCD{\tfrac{\tau+\knd+1}{\sqrt{\amn{2}\blx}}}-\GCD{\tfrac{\tau+\knd}{\sqrt{\amn{2}\blx}}}
	+2\tfrac{0.56}{\sqrt{\blx}} 
	\tfrac{\amn{3}}{\amn{2}\sqrt{\amn{2}}}
	\\
	\notag 
	&\leq \tfrac{1}{\sqrt{\amn{2}\blx}}\left(\tfrac{1}{\sqrt{2\pi}}+2\tfrac{0.56 \amn{3}}{\amn{2}}\right).
	\\
\notag
	&\leq e\htdelta\blx^{-\sfrac{1}{2}}.
	\end{align}
	Thus we can bound \(\mW(\oev \cap \set{B}_{\integers{+}})\)
	using \eqref{eq:SC-HTBE-CB2}, the fact that 
	\(
	\wma{\rno}{\mQ}(\oev \cap \set{B}_{\knd})\!\leq\!\wma{\rno}{\mQ}(\set{B}_{\knd})\), 
	and the formula for the sum of a geometric series, as well:
	\begin{align}
	\notag
	\mW(\oev \cap \set{B}_{\integers{+}})
	&=\sum\nolimits_{\knd\in\integers{+}} \mW(\oev \cap \set{B}_{\knd})
	\\
	\notag
	&\leq\sum\nolimits_{\knd\in\integers{+}}  
	e\htdelta\blx^{-\sfrac{1}{2}}
	e^{-\RD{1}{\wma{\rno}{\mQ}}{\mW}+(1-\rno)\tau+(1-\rno)\knd}
	\\
	\label{eq:SC-HTBE-converse-3}
	&\leq
	e\htdelta\blx^{-\sfrac{1}{2}}	
	e^{-\RD{1}{\wma{\rno}{\mQ}}{\mW}+(1-\rno)\tau}\tfrac{e^{1-\rno}}{1-e^{1-\rno}}.
	\end{align}
	For \(\tau\!=\!\tfrac{2\ln \htdelta-2\ln\rnb-\ln\blx}{2\rno}+\tfrac{1}{\rno}\ln \tfrac{e-1}{e^{\rno-1}-1}-1\),
	\eqref{eq:lem:SC-HTBE-converse} follows from
	\eqref{eq:SC-HTBE-converse-1},
	\eqref{eq:SC-HTBE-converse-2}, \eqref{eq:SC-HTBE-converse-3}, and 
	the identity \(\tfrac{(e-1)^{1-\rno}}{e^{\rno-1}-1}\leq \tfrac{1}{\rno-1}\).

\end{proof}

\section{Augustin Information, Augustin Mean and The Strong Converse Exponent}\label{sec:preliminary}
Our primary goal in this section is to define 
the Augustin information and mean and 
the strong converse exponent 
and review those properties of them that 
will be useful in our analysis. 
Let us start by defining the conditional 
\renyi divergence: 
\begin{definition}\label{def:conditionaldivergence}
	For any \(\rno\in\reals{+}\), \(\Wm:\inpS\to\pmea{\outA}\), \(\mQ\in\pmea{\outA}\),
	and \(\mP\in\pdis{\inpS}\) \emph{the order \(\rno\) conditional \renyi divergence for 
		the input distribution \(\mP\)} is
	\begin{align}
	\label{eq:def:conditionaldivergence}
	\CRD{\rno}{\Wm\!}{\mQ}{\mP}
	&\DEF \sum\nolimits_{\dinp\in \inpS}  \mP(\dinp) \RD{\rno}{\Wm(\dinp)}{\mQ}.
	\end{align}
\end{definition}

\begin{definition}\label{def:information}
\hspace{-.2cm} For any \(\!\rno\!\in\!\reals{+}\), \(\!\Wm\!:\!\inpS\!\to\!\pmea{\outA}\), and 
\(\!\mP\!\in\!\pdis{\inpS}\!\) 
\emph{the order \(\rno\!\) Augustin information for the input distribution \(\mP\!\)} is
\begin{align}
\label{eq:def:information}
\RMI{\rno}{\mP}{\Wm}
&\DEF \inf\nolimits_{\mQ\in \pmea{\outA}} \CRD{\rno}{\Wm\!}{\mQ}{\mP}.
\end{align}
\end{definition}
The infimum in \eqref{eq:def:information} is achieved by a unique probability 
measure \(\qmn{\rno,\mP}\), called
\emph{the order \(\rno\) Augustin mean for the input distribution \(\mP\)},
by \cite[{Lemma \ref*{C-lem:information}-(\ref*{C-information:one},\ref*{C-information:zto},\ref*{C-information:oti})}]{nakiboglu19C}. 
Furthermore,
\begin{align}
\label{eq:augustinslaw}
\hspace{-.35cm}\RD{1 \vee \rno}{\qmn{\rno,\mP}}{\mQ}
\!\geq\!
\CRD{\rno}{\Wm\!}{\mQ}{\mP}\!-\!\RMI{\rno}{\mP}{\Wm}
&\!\geq\!
\RD{1 \wedge \rno}{\qmn{\rno,\mP}}{\mQ}
\end{align}
for all \(\mQ\in\pmea{\outA}\) by 
\cite[{Lemma \ref*{C-lem:information}-(\ref*{C-information:one},\ref*{C-information:zto},\ref*{C-information:oti})}]{nakiboglu19C}, as well.
\(\RMI{\rno}{\mP}{\Wm}\) is continuously differentiable in \(\rno\) on \(\reals{+}\)
and 
\begin{align}
\label{eq:lem:informationO:differentiability-alt}
\pder{}{\rno}\RMI{\rno}{\mP}{\Wm}
&=\begin{cases}
\tfrac{1}{(\rno-1)^2}\CRD{1}{\Wma{\rno}{\qmn{\rno,\mP}}}{\Wm}{\mP}
&\rno\neq 1
\\
\sum\nolimits_{\dinp}\tfrac{\mP(\dinp)}{2}
\VXS{\Wm(\dinp)}{\ln \der{\Wm(\dinp)}{\qmn{1,\mP}}}
&\rno= 1
\end{cases}
\end{align}
by \cite[Lemma {\ref*{C-lem:informationO}-(\ref*{C-informationO:differentiability})}]{nakiboglu19C},
where \(\Wma{\rno}{\qmn{\rno,\mP}}(\dinp)\) 
is the order \(\rno\) tilted probability measure between \(\!\Wm\!(\dinp)\!\) and \(\qmn{\rno,\mP}\).
 
\(\Wma{\rno}{\qmn{\rno,\mP}}\) is called 
\emph{the order \(\rno\) tilted channel} for the channel \(\Wm\)
and the output distribution \(\qmn{\rno,\mP}\).
The tilted channel is also used to express \(\RMI{\rno}{\mP}{\Wm}\)
in terms of the Kullback-Leibler divergences in
\cite[Lemma \ref*{C-lem:information}-(\ref*{C-information:alternative})]{nakiboglu19C}:
\begin{align}
\label{eq:lem:information:alternative:opt}
\RMI{\rno}{\mP}{\Wm}
&=\tfrac{\rno}{1-\rno}\CRD{1}{\Wma{\rno}{\qmn{\rno,\mP}}}{\Wm}{\mP}
+\RMI{1}{\mP}{\Wma{\rno}{\qmn{\rno,\mP}}}.
\end{align}
Since \(\sum\nolimits_{\dinp}\!\mP(\dinp)\!\Wma{\rno}{\qmn{\rno,\mP}}\!(\dinp)
\!=\!\qmn{\rno,\mP}\!\)
by\footnote{In fact the Augustin mean is the only probability measure 
	satisfying such a fixed point property by
	\cite[{Lemma \!\ref*{C-lem:information}}]{nakiboglu19C},
	as well.} 
\cite[{Lemma \ref*{C-lem:information}-(\ref*{C-information:one},\ref*{C-information:zto},\ref*{C-information:oti})}]{nakiboglu19C},
we also have the following identity for all \(\rno\!\in\!\reals{+}\)
\begin{align}
\label{eq:haroutunian}
\RMI{1}{\mP}{\Wma{\rno}{\qmn{\rno,\mP}}}
&=\CRD{1}{\Wma{\rno}{\qmn{\rno,\mP}}}{\qmn{\rno,\mP}}{\mP}.
\end{align}
A more comprehensive discussion of Augustin's information measures
can be found in \cite{nakiboglu19C}. 

\begin{definition}\label{def:strongconverseexponent}
	For any \(\Wm\!:\!\inpS\to \pmea{\outA}\), \(\mP\!\in\!\pdis{\inpS}\), 
	and \(\rate\!\in\!\reals{+}\), the strong converse exponent (SCE)  is
	\begin{align}
	\label{eq:def:strongconverseexponent}
	\sce{\rate,\Wm\!,\mP}
	&\DEF \sup\nolimits_{\rno\in (1,\infty)} \tfrac{1-\rno}{\rno} \left(\RMI{\rno}{\mP}{\Wm}-\rate\right).
	\end{align}
\end{definition}
We can apply the derivative test to determine \(\sce{\rate,\Wm\!,\mP}\),
because \(\RMI{\rno}{\mP}{\Wm}\) is continuously differentiable in
the order \(\rno\) by 
\cite[Lemma  {\ref*{C-lem:informationO}-(\ref*{C-informationO:differentiability})}]{nakiboglu19C}. 
Equations \eqref{eq:lem:informationO:differentiability-alt} and
\eqref{eq:lem:information:alternative:opt}
imply 
\begin{align}
\label{eq:parametric-derivative}
\pder{}{\rno}\tfrac{1-\rno}{\rno}\left(\RMI{\rno}{\mP}{\Wm}\!-\!\rate\right)
&\!=\!\tfrac{1}{\rno^{2}}\!\left(\rate\!-\!
\RMI{1}{\mP}{\Wma{\rno}{\qmn{\rno,\mP}}}\right).
\end{align}
On the other hand, either 
\(\RMI{1}{\mP}{\Wma{\rno}{\qmn{\rno,\mP}}}\) 
is increasing and continuous in \(\rno\) on \(\reals{+}\),
or
\(\RMI{1}{\mP}{\Wma{\rno}{\qmn{\rno,\mP}}}\!=\!\RMI{1}{\mP}{\Wm}\) 
for all positive \(\rno\)
by 
\cite[Lemma {\ref*{C-lem:informationO}-(\ref*{C-informationO:monotonicityofharoutunianinformation})}]{nakiboglu19C}.
Furthermore,
\(\RMI{1}{\mP}{\Wma{1}{\qmn{1,\mP}}}\) 
is equal to  \(\RMI{1}{\mP}{\Wm}\).
Thus
for any rate \(\rate\) in 
\((\RMI{1}{\mP}{\Wm},\lim_{\rno\uparrow \infty\!}\RMI{1}{\mP}{\Wma{\rno}{\qmn{\rno,\mP}}})\),
there exists an order \(\rns\) in \((1,\!\infty\!)\) satisfying 
\begin{align}
\label{eq:parametric-haroutunianform-rate}
\rate
&=\RMI{1}{\mP}{\Wma{\rns}{\qmn{\rns,\mP}}}
\end{align}
by the intermediate value theorem \cite[4.23]{rudin}.
The \(\rns\) satisfying \eqref{eq:parametric-haroutunianform-rate}
is unique  because \(\RMI{1}{\mP}{\Wma{\rno}{\qmn{\rno,\mP}}}\)
is increasing in \(\rno\). 
The monotonicity of
\(\RMI{1}{\mP}{\Wma{\rno}{\qmn{\rno,\mP}}}\) in \(\rno\)
and  \eqref{eq:parametric-derivative} also implies
\(\sce{\rate,\Wm\!,\mP}=\tfrac{1-\rns}{\rns}\left(\RMI{\rns}{\mP}{\Wm}\!-\!\rate\right)\).
Thus as a result of  \eqref{eq:lem:information:alternative:opt}, 
the unique \(\rns\) satisfying \eqref{eq:parametric-haroutunianform-rate} also satisfies
\begin{align}
\label{eq:parametric-haroutunianform-exponent}
\sce{\rate,\Wm\!,\mP}
&=\CRD{1}{\Wma{\rns}{\qmn{\rns,\mP}}}{\Wm}{\mP}.
\end{align} 

Since \(\CRD{1}{\Wma{\rno}{\qmn{\rno,\mP}}}{\qmn{\rno,\mP}}{\mP}\) 
is continuous and increasing in \(\rno\), 
its inverse is increasing  and continuous, as well. 
Thus the definition of 
SCE given in \eqref{eq:def:strongconverseexponent} 
and the definition of derivative as a limit
imply that for any \(\rate\) in 
\(\left(\RMI{1}{\mP}{\Wm},\lim_{\rno \uparrow \infty}\RMI{1}{\mP}{\Wma{\rno}{\qmn{\rno,\mP}}}\right)\)
the unique \(\rns\)
satisfying \eqref{eq:parametric-haroutunianform-rate}
also satisfies
\begin{align}
\label{eq:parametric-haroutunianform-slope}
\pder{}{\rate}\sce{\rate,\Wm\!,\mP}
&=\tfrac{\rns-1}{\rns}.
\end{align}
If \(\rate\geq \lim_{\rno \uparrow \infty}\RMI{1}{\mP}{\Wma{\rno}{\qmn{\rno,\mP}}}\),
then the derivative given in
\eqref{eq:parametric-derivative}
is positive for all \(\rno\in(1,\infty)\) and
thus
\begin{align}
\notag
\sce{\rate,\Wm\!,\mP}
&=\lim\nolimits_{\rno\uparrow \infty} \tfrac{1-\rno}{\rno} \left(\RMI{\rno}{\mP}{\Wm}-\rate\right)
\\
\label{eq:highrate-SCE}
&=\rate-\RMI{\infty}{\mP}{\Wm} 
\end{align}
for all \(\rate\geq \lim_{\rno \uparrow \infty}\RMI{1}{\mP}{\Wma{\rno}{\qmn{\rno,\mP}}}\).

On the other hand, if \(\rate\leq \RMI{1}{\mP}{\Wma{1}{\qmn{1,\mP}}}\),
then the derivative given in
\eqref{eq:parametric-derivative}
is negative for all \(\rno\in(1,\infty)\) and
thus
\begin{align}
\notag
\sce{\rate,\Wm\!,\mP}
&=\lim\nolimits_{\rno\downarrow 1} \tfrac{1-\rno}{\rno} \left(\RMI{\rno}{\mP}{\Wm}-\rate\right)
\\
\label{eq:low-SCE}
&=0.
\end{align}
for all \(\rate\leq  \RMI{1}{\mP}{\Wm}\).

Equations 
\eqref{eq:parametric-haroutunianform-rate},
\eqref{eq:parametric-haroutunianform-exponent},
\eqref{eq:parametric-haroutunianform-slope},
\eqref{eq:highrate-SCE}, and \eqref{eq:low-SCE}
characterize the strong converse exponent 
\(\sce{\rate,\Wm\!,\mP}\) defined in 
\eqref{eq:def:strongconverseexponent} 
as a non-decreasing continuously differentiable convex function 
that is 
strictly convex on 
\((\RMI{1}{\mP}{\Wm},\lim_{\rno \uparrow \infty}\RMI{1}{\mP}{\Wma{\rno}{\qmn{\rno,\mP}}})\)
and 
increasing on \((\RMI{1}{\mP}{\Wm},\infty)\).

\begin{remark}
The definition of \(\sce{\rate,\Wm\!,\mP}\) given in \eqref{eq:def:strongconverseexponent} is equivalent to
the one used by Dueck and \korner \cite{dueckK79}.
In order to see why, recall that the Augustin
information satisfies the following variational
characterization by \cite[Lemma \ref*{C-lem:information}-(\ref*{C-information:alternative})]{nakiboglu19C}
\begin{align}
\notag
\tfrac{1-\rno}{\rno}\RMI{\rno}{\mP}{\Wm}
&=\inf\limits_{\Vm\in\pmea{\outA|\inpS}}\CRD{1}{\Vm}{\Wm}{\mP}+\tfrac{1-\rno}{\rno}\RMI{1}{\mP}{\Vm}.
\end{align}
Thus \(\sce{\rate,\Wm\!,\mP}\) can be written as follows 
for \(\mS=\tfrac{\rno-1}{\rno}\):
\begin{align}
\notag
&\hspace{-.3cm}\sce{\rate,\Wm,\mP}
\\
\notag
&=\sup_{\mS\in(0,1)} 
\inf\limits_{\Vm\in\pmea{\outA|\inpS}}\CRD{1}{\Vm}{\Wm}{\mP}+\mS(\rate-\RMI{1}{\mP}{\Vm}) 
\\
\notag
&=\inf\limits_{\Vm\in\pmea{\outA|\inpS}}\sup_{\mS\in(0,1)} 
\CRD{1}{\Vm}{\Wm}{\mP}+\mS(\rate-\RMI{1}{\mP}{\Vm}) 
\\
\notag
&=\inf\limits_{\Vm\in\pmea{\outA|\inpS}}
\CRD{1}{\Vm}{\Wm}{\mP}+\abp{\rate-\RMI{1}{\mP}{\Vm}}.
\end{align}
We can change the order of the infimum and supremum 
using Sion's minimax theorem \cite{sion58,komiya88}
because we can replace \(\pmea{\outA|\inpS}\)
by the set of elements of 
\(\pmea{\outA|\supp{\mP}}\) satisfying
\(\CRD{1}{\Vm}{\Wm}{\mP}\leq \rate\) and 
the latter set is compact in the topology of setwise 
convergence by the necessary and sufficient condition 
for the uniform integrability given by de la
Vallee Poussin \cite[Thm. 4.5.9]{bogachev},
see \cite[(d-iii) on p.36]{nakiboglu19C} for a
similar argument.
\end{remark}

\section{The Refined Strong Converse}\label{sec:RSC}
\begin{theorem}\label{thm:constantcomposition-RSC-lowrate}
For any \(\Wm\!:\!\inpS\!\to\!\pmea{\outA}\), 
\(M,L,\blx\!\in\!\integers{+}\),
\(\mP\!\in\!\pdis{\inpS}\) satisfying 
\(\RMI{1}{\mP}{\Wm}\!<\!\tfrac{1}{\blx}\ln\tfrac{M}{L}
\!<\!\lim\nolimits_{\rno \uparrow \infty}\RMI{1}{\mP}{\Wma{\rno}{\qmn{\rno,\mP}}}\)
and 
\(\blx\mP(\dinp)\!\in\!\integers{\geq0}\) for all \(\dinp\!\in\!\inpS\),
the order \(\rns\!\DEF\!\tfrac{1}{1-\dsce{\frac{1}{\blx}\ln\frac{M}{L},\Wm,\mP}}\) 
	satisfies
	\begin{align}
\label{eq:thm:constantcomposition-RSC-lowrate-hypothesis}
\RMI{1}{\mP}{\Wma{\rns}{\qmn{\rns,\mP}}}
&=\tfrac{1}{\blx}\ln \tfrac{M}{L}.
\end{align}
Furthermore, any \((M,L)\) channel code of length \(\blx\)
whose codewords all have the same composition \(\mP\) satisfies
\begin{align}
\label{eq:thm:constantcomposition-RSC-lowrate}
\Pem{(\blx)}
&\geq
1
-2e^{\rns}\left(\tfrac{\htdelta}{\rns-1}\right)^{\frac{1}{\rns}}
\blx^{-\sfrac{1}{2\rns}}
e^{-\blx \sce{\frac{1}{\blx}\ln\frac{M}{L},\Wm,\mP}}	
	\end{align}	
where
	\begin{align}
	\label{eq:thm:constantcomposition-RSC-lowrate-a2}
	\amn{2}&=
	\EXS{\mP\mtimes\Wma{\rns}{\qmn{\rns,\mP}}}{\abs{\ln\!\der{\Wm}{\qmn{\rns,\mP}}-\EXS{\Wma{\rns}{\qmn{\rns,\mP}}}{\ln\!\der{\Wm}{\qmn{\rns,\mP}}}}^{2}},
	\\
	\label{eq:thm:constantcomposition-RSC-lowrate-a3}
	\amn{3}&=
	\EXS{\mP\mtimes\Wma{\rns}{\qmn{\rns,\mP}}}{
		\abs{\ln\!\der{\Wm}{\qmn{\rns,\mP}}-\EXS{\Wma{\rns}{\qmn{\rns,\mP}}}{\ln\!\der{\Wm}{\qmn{\rns,\mP}}}
		}^{3}},
	\\
\label{eq:thm:constantcomposition-RSC-lowrate-htdelta}
\htdelta
&\!\DEF\! \tfrac{1}{e\sqrt{\amn{2}}}\left(\tfrac{1}{\sqrt{2\pi}}+2\tfrac{0.56\amn{3}}{\amn{2}}\right).
\end{align}
\end{theorem}
\begin{theorem}\label{thm:constantcomposition-RSC-highrate}
	For any \(\Wm\!:\!\inpS\!\to\!\pmea{\outA}\), 
	\(M,L,\blx\!\in\!\integers{+}\),
	\(\mP\!\in\!\pdis{\inpS}\) satisfying 
	\(\tfrac{1}{\blx}\ln\tfrac{M}{L}\!\geq\!\lim\nolimits_{\rno \uparrow \infty}\RMI{1}{\mP}{\Wma{\rno}{\qmn{\rno,\mP}}}\)
	and 
	\(\blx\mP(\dinp)\!\in\!\integers{\geq0}\) for all \(\dinp\!\in\!\inpS\),
	any \((M,L)\) channel code of length \(\blx\)
	whose codewords all have the same composition \(\mP\) satisfies
	\begin{align}
	\label{eq:thm:constantcomposition-RSC-highrate}
	\Pem{(\blx)}
	&\geq
	1-e^{-\blx \sce{\frac{1}{\blx}\ln\frac{M}{L},\Wm,\mP}}.
	\end{align}	
\end{theorem}
Theorems 
\ref{thm:constantcomposition-RSC-lowrate}
and
\ref{thm:constantcomposition-RSC-highrate}
collectively imply 
for  all \(\blx\!\in\!\integers{+}\) 
a strong converse of the form
\eqref{eq:RSC} for 
\(\sce{\rate}=\sce{\tfrac{1}{\blx}\ln\tfrac{M}{L},\Wm,\mP}\),
for a constant \(A\) 
determined by the rate \(\rate\), 
the channel \(\Wm\!\), 
and the composition \(\mP\).
Following the convention used for the corresponding improvement of the
sphere packing bound in \cite{altugW11,altugW14A,chengHT19,nakiboglu19-ISIT,nakiboglu20F},
we call these bounds refined strong converses.
\begin{proof}[Proof of Theorem \ref{thm:constantcomposition-RSC-lowrate}]
	The existence of a unique order \(\rns\) satisfying
	\eqref{eq:thm:constantcomposition-RSC-lowrate-hypothesis} was 
	proved and its value was determined  in \S\ref{sec:preliminary}, 
	see \eqref{eq:parametric-haroutunianform-rate},
	\eqref{eq:parametric-haroutunianform-exponent},
	and
	\eqref{eq:parametric-haroutunianform-slope}.
	
	Let the probability measures \(\wmn{\dmes}\), \(\mQ\),  and \(\vmn{\dmes}\)
	in \(\pmea{\outA_{1}^{\blx}}\) be
	\begin{align}
	\notag
	\wmn{\dmes} 
	&\DEF \bigotimes\nolimits_{\tin=1}^{\blx} \Wm(\enc_{\tin}(\dmes)),
	\\
	\notag
	\mQ
	&\DEF \bigotimes\nolimits_{\tin=1}^{\blx} \qmn{\rns,\mP},
	\\
	\notag
	\vmn{\dmes} 
	&\DEF \bigotimes\nolimits_{\tin=1}^{\blx} \Wma{\rns}{\qmn{\rns,\mP}}(\enc_{\tin}(\dmes)).
	\end{align}
	Then \(\vmn{\dmes}\) is equal to the order \(\rns\)
	tilted probability measure between \(\wmn{\dmes}\) and \(\mQ\).
	Furthermore,\footnote{It is worth mentioning that
		both
		\(\RD{1}{\vmn{\dmes}}{\mQ}\) and \(\RD{1}{\vmn{\dmes}}{\wmn{\dmes}}\)	
		can be expressed in this form for all messages because all
		\(\enc(\dmes)\)'s have the same composition \(\mP\).} 
	\begin{align}
	\notag
	\RD{1}{\vmn{\dmes}}{\mQ}
	&=\blx\CRD{1}{\Wma{\rns}{\qmn{\rns,\mP}}}{\qmn{\rns,\mP}}{\mP}
	&
	&\dmes\in\mesS, 
	\\
	\notag
	\RD{1}{\vmn{\dmes}}{\wmn{\dmes}}
	&=\blx\CRD{1}{\Wma{\rns}{\qmn{\rns,\mP}}}{\Wm}{\mP}
	&
	&\dmes\in\mesS. 
	\end{align}

Note that
\(\CRD{1}{\Wma{\rns}{\qmn{\rns,\mP}}}{\qmn{\rns,\mP}}{\mP}
\!=\!\frac{1}{\blx}\ln\frac{M}{L}\)
by \eqref{eq:haroutunian} and \eqref{eq:thm:constantcomposition-RSC-lowrate-hypothesis}
and 
\(\CRD{1}{\Wma{\rns}{\qmn{\rns,\mP}}}{\Wm}{\mP}
\!=\!\sce{\frac{1}{\blx}\ln\frac{M}{L},\Wm,\mP}\)
by 
\eqref{eq:parametric-haroutunianform-rate},
\eqref{eq:parametric-haroutunianform-exponent},
and \eqref{eq:thm:constantcomposition-RSC-lowrate-hypothesis}.
Thus applying Lemma \ref{lem:SC-HTBE-converse}, for 
\(\oev=\{\dout_{1}^{\blx}:\dmes\!\in\!\dec(\dout_{1}^{\blx})\}\)
and \(\rnb=\mQ(\dmes\!\in\!\dec)\tfrac{M}{L}\)
we get
\begin{align}
\label{eq:constantcomposition-1}
\Pem{\dmes}
&\!\geq\!
1\!-\!\tfrac{2e^{\rns}\htdelta^{\sfrac{1}{\rns}}}{(\rns-1)^{\sfrac{1}{\rns}}}
\left( \tfrac{\mQ(\dmes\in\dec) M}{L}\right)^{\frac{\rns-1}{\rns}}
\tfrac{e^{-\blx \sce{\frac{1}{\blx}\ln\frac{M}{L},\Wm,\mP}}}{\blx^{\sfrac{1}{2\rns}}}.	
\end{align}
On the other hand
\(\sum_{m\in\mesS}\mQ(\dmes\in\dec)\leq L\),
as a result of the definition of the list decoding. 
Thus using the concavity of the function
\(\dsta^{\frac{\rns-1}{\rns}}\) in \(\dsta\)
together with the Jensen's inequality get
\begin{align}
\notag
\sum\nolimits_{m\in\mesS} \tfrac{1}{M}
\left( \tfrac{\mQ(\dmes\in\dec) M}{L}\right)^{\frac{\rns-1}{\rns}}
&\leq \left(\sum\nolimits_{m\in\mesS} \tfrac{1}{M} \tfrac{ \mQ(\dmes\in\dec) M}{L}\right)^{\frac{\rns-1}{\rns}}
\\
\notag
&=1.
\end{align}
Then \eqref{eq:thm:constantcomposition-RSC-lowrate}
follows from \eqref{eq:constantcomposition-1}
and the definition error probability as the
average of the conditional error probabilities.
\end{proof}

\begin{proof}[Proof of Theorem \ref{thm:constantcomposition-RSC-highrate}]
Let the probability measures \(\wmn{\dmes}\), \(\qmn{\rno}\)  be
\begin{align}
\notag
\wmn{\dmes} 
&\DEF \bigotimes\nolimits_{\tin=1}^{\blx} \Wm(\enc_{\tin}(\dmes)),
&&\mbox{and}&
\qmn{\rno}
&\DEF \bigotimes\nolimits_{\tin=1}^{\blx} \qmn{\rno,\mP}.
\end{align}
Then \(\RD{\rno}{\wmn{\dmes}}{\qmn{\rno}}=\blx\RMI{\rno}{\mP}{\Wm}\)
for all \(\dmes\) because all \(\enc(\dmes)\)'s have the composition \(\mP\).
On the other hand the data processing inequality of the \renyi divergence,
\cite[Thm 9]{ervenH14}, imply
\begin{align}
\notag
\RD{\rno}{\wmn{\dmes}}{\qmn{\rno}}
&\geq \tfrac{\ln \left[
	(\Pem{\dmes})^{\rno}(\qmn{\rno}(\dmes\in\dec))^{1-\rno}
	+(1-\Pem{\dmes})^{\rno}(\qmn{\rno}(\dmes\in\dec))^{1-\rno}
	\right]}{\rno-1}
\\
\notag
&\geq \tfrac{\ln \left[(1-\Pem{\dmes})^{\rno}(\qmn{\rno}(\dmes\in\dec))^{1-\rno}
	\right]}{\rno-1}.
\end{align}
Thus \(\Pem{\dmes}
\!\geq\!
1-\left(\qmn{\rno}(\dmes\in\dec)\right)^{\frac{\rno-1}{\rno}}e^{\frac{\rno-1}{\rno}\blx\RMI{\rno}{\mP}{\Wm}} 
\).
On the other hand 
the concavity of the function \(\dsta^{\frac{\rno-1}{\rno}}\) in \(\dsta\)
for \(\rno>1\), the Jensen's inequality, and
\(\sum_{m\in\mesS}\qmn{\rno}(\dmes\in\dec)\!\leq\!L\),
imply
\(\sum\nolimits_{m\in\mesS} \tfrac{1}{M} \left(\qmn{\rno}(\dmes\in\dec)\right)^{\frac{\rno-1}{\rno}}
\leq \left(\sfrac{L}{M}\right)^{\frac{\rno-1}{\rno}}\). Hence
\begin{align}
\notag
\Pem{(\blx)}
&\geq
1-e^{-\frac{1-\rno}{\rno}\blx \left(\RMI{\rno}{\mP}{\Wm}-\frac{1}{\blx}\ln \frac{M}{L}\right)} 
&
&\forall \rno\in(1,\infty).
\end{align}
Then \eqref{eq:thm:constantcomposition-RSC-highrate} follows from \eqref{eq:def:strongconverseexponent}.
\end{proof}

\section{Discussion}\label{sec:conclusion}
Although we have confined our analysis to the constant composition 
codes for brevity,
using the Augustin capacity and center 
---instead of Augustin information and mean---
one can obtain analogous results for 
additive white Gaussian noise channels with quadratic cost functions
and
\renyi symmetric channels defined in \cite{nakiboglu20F}.
For \renyi symmetric channels the refined 
strong converse \eqref{eq:RSC}, 
can be established with smaller, i.e., better, constant \(A\)
using the saddle point approximation.
Such a result has been reported in \cite[(36)]{vazquezFKL18},
assuming a common support for all output distributions of 
the channel and a non-lattice structure for the random 
variables involved.
Establishing refined strong converses without any symmetry hypothesis
is the main technical challenge in this line of work.

We believe 
the refined strong converses of the form \eqref{eq:RSC} 
are the best possible bounds 
for derivations of the strong converse relying on 
the asymptotic behavior of sums of independent 
random variables.
Nevertheless for the singular symmetric channels considered 
in \cite{altugW19}, it should be possible to 
improve \eqref{eq:RSC} as 
\(\Pem{(\blx)}\geq 1- A n^{-0.5} e^{-\blx\sce{\rate}}\).

 \appendices
\section{An Matching Achievability Result For Lemma \ref{lem:SC-HTBE-converse}}\label{sec:SC-HTBE-achievability}
\begin{lemma}\label{lem:SC-HTBE-achievability}
	For any \(\rno\!\in\!(1,\infty)\),
	\(\blx\!\in\!\integers{+}\), 
	\(\wmn{\tin},\qmn{\tin}\!\in\!\pmea{\outA_{\tin}}\), 
	let \(\wmn{\tin,ac}\) be the component of \(\wmn{\tin}\) that is absolutely continuous in \(\qmn{\tin}\)
	and let \(\amn{2}\), \(\amn{3}\), and \(\htdelta\) be
	\begin{align}
	\notag
	\amn{2}
	&\!\DEF\!\tfrac{1}{\blx}\!\sum\nolimits_{\tin=1}^{\blx} 
	\EXS{\wma{\rno}{\mQ}}{\left(\ln\!\der{\wmn{\tin,ac}}{\qmn{\tin}}-\EXS{\wma{\rno}{\mQ}}{\ln\!\der{\wmn{\tin,ac}}{\qmn{\tin}}}\right)^{2}},
	\\
	\notag
	\amn{3}
	&\!\DEF\!\tfrac{1}{\blx}\!\sum\nolimits_{\tin=1}^{\blx} 
	\EXS{\wma{\rno}{\mQ}}{\abs{\ln\!\der{\wmn{\tin,ac}}{\qmn{\tin}}-\EXS{\wma{\rno}{\mQ}}{\ln\!\der{\wmn{\tin,ac}}{\qmn{\tin}}}}^{3}},
	\\
	\notag
	\htdelta
	&\!\DEF\! \tfrac{1}{e\sqrt{\amn{2}}}\left(\tfrac{1}{\sqrt{2\pi}}+2\tfrac{0.56\amn{3}}{\amn{2}}\right),
	\end{align}
	where \(\mW\!=\!\otimes_{\tin=1}^{\blx} \wmn{\tin}\) and \(\mQ\!=\!\otimes_{\tin=1}^{\blx} \qmn{\tin}\).
	Then for any \(\rno\!\in\!(1,\infty)\) and 
	\(\rnb\!\in\!\left[\tfrac{9\htdelta e^{\rno e\sqrt{2\pi e\amn{2}}\htdelta}}{\sqrt{\blx}}
	e^{-\rno\sqrt{\amn{2}\blx}},
	\tfrac{9\htdelta}{\sqrt{\blx}} e^{\rno\sqrt{\amn{2}\blx}}\right]\),
	there exists an \(\oev\!\in\!\outA_{1}^{\blx}\) such that
	\begin{subequations}
		\label{eq:lem:SC-HTBE-achievability}
		\begin{align}
		\label{eq:lem:SC-HTBE-achievability-q}
		\mQ(\oev)
		&\!\leq\!\rnb e^{-\RD{1}{\wma{\rno}{\mQ}}{\mQ}},	
		\\
		\notag
\hspace{-.2cm}
		\mW({\outS_{1}^{\blx}}\setminus\oev)
		&\!\leq\! \left(\prod\limits_{\tin=1}^{\blx}\lon{\wmn{\tin,ac}}\right)\\
		\label{eq:lem:SC-HTBE-achievability-w}
		&~~~~~~~~~
		-\tfrac{e^{(1-\rno)e\sqrt{2\pi e\amn{2}}\htdelta}}{\sqrt{2\pi\amn{2}}}
		(\tfrac{\rnb}{9\htdelta})^{\frac{\rno-1}{\rno}}
		\tfrac{e^{-\RD{1}{\wma{\rno}{\mQ}}{\mW}}}{\blx^{\sfrac{1}{2\rno}}}
		\end{align}
	\end{subequations}		
\end{lemma}
\begin{proof}[Proof of Lemma \ref{lem:SC-HTBE-achievability}]
	Let the event \(\oev\) be 
	\begin{align}
	\notag
	\oev
	&=\left\{\dout_{1}^{\blx}:
	\der{\mW}{\rfm}>0 \mbox{~and~} \der{\mQ}{\rfm}=0\right\}
	\bigcup \cup_{\knd\in\integers{\geq0}} \set{B}_{\knd}^{\delta}
	\end{align}
	where \(\rfm\) is any reference measure satisfying \(\mW\AC\rfm\) and \(\mQ\AC\rfm\)
	and the event \(\set{B}_{\knd}^{\delta}\) is defined for each integer \(\knd\) as
\begin{align}
\label{eq:SC-HTBE-Bdef-delta}
\set{B}_{\knd}^{\delta}
&\!\DEF\! \left\{\dout_{1}^{\blx}\!:\!\tau+\knd{\delta}\leq \cln{}-\EXS{\wma{\rno}{\mQ}}{\cln{}}<\tau+(\knd+1){\delta} \right\}. 
\end{align}
Note that 
\(\set{B}_{\integers{}}=\sum_{\knd\in\integers{}}\set{B}_{k}^{\delta}\)
where \(\set{B}_{\integers{}}\DEF\sum_{\knd\in\integers{}}\set{B}_{k}\)
for \(\set{B}_{\knd}\) defined in \eqref{eq:SC-HTBE-Bdef},
which is equal to  \(\set{B}_{\knd}^{1}\).
Then
\begin{align}
\label{eq:SC-HTBE-AB1}
\mQ(\oev)
&=\sum\nolimits_{\knd\in\integers{\geq0}}\mQ(\set{B}_{\knd}^{\delta}),
\\
\mW(\outS_{1}^{\blx}\setminus\oev)
\label{eq:SC-HTBE-AB2}
&=\mW(\set{B}_{\integers{}})-\sum\nolimits_{\knd\in\integers{\geq0}}\mW(\set{B}_{\knd}^{\delta}).
\end{align}
On the other hand, as a result of \eqref{eq:CLQ} and \eqref{eq:CLW}, we have
	\begin{align}
	\label{eq:SC-HTBE-AB3}
	\mQ(\set{B}_{\knd}^{\delta})
	&\leq
	\wma{\rno}{\mQ}(\set{B}_{\knd}^{\delta}) e^{-\RD{1}{\wma{\rno}{\mQ}}{\mQ}-\rno\tau-\rno \knd\delta},
	\\
	\label{eq:SC-HTBE-AB4}
	\mW(\set{B}_{\knd}^{\delta})
	&\geq \wma{\rno}{\mQ}(\set{B}_{\knd}^{\delta})
	e^{-\RD{1}{\wma{\rno}{\mQ}}{\mW}+(1-\rno)\tau+(1-\rno)(\knd+1){\delta}}.
	\end{align}
Since \(\cln{\tin}\)'s are jointly independent under the probability measure \(\wma{\rno}{\mQ}\),
we can bound \(\wma{\rno}{\mQ}(\set{B}_{\knd}^{\delta})\)
from above  using the Berry-Esseen theorem
\cite{berry41,esseen42,shevtsova10}:
	\begin{align}
	\notag
	\wma{\rno}{\mQ}(\set{B}_{\knd}^{\delta})
	&\leq \GCD{\tfrac{\tau+(\knd+1)\delta}{\sqrt{\amn{2}\blx}}}-\GCD{\tfrac{\tau+\knd \delta}{\sqrt{\amn{2}\blx}}}
	+2\tfrac{0.56}{\sqrt{\blx}} 
	\tfrac{\amn{3}}{\amn{2}\sqrt{\amn{2}}}
	\\
	\notag 
	&\leq \tfrac{1}{\sqrt{\amn{2}\blx}}\left(\tfrac{\delta}{\sqrt{2\pi}}+2\tfrac{0.56 \amn{3}}{\amn{2}}\right).
	\end{align}
Thus using first \eqref{eq:SC-HTBE-AB1}, \eqref{eq:SC-HTBE-AB3}, 
	and the formula for the sum of geometric series we get
	\begin{align}
	\notag
	\mQ(\oev)
	&\leq e^{-\RD{1}{\wma{\rno}{\mQ}}{\mQ}-\rno\tau}
	\sum\nolimits_{\knd\in\integers{\geq0}} \wma{\rno}{\mQ}(\set{B}_{\knd}^{\delta}) e^{-\rno\knd \delta}
	\\
	\label{eq:SC-HTBE-AB5}
	&\leq e^{-\RD{1}{\wma{\rno}{\mQ}}{\mQ}-\rno \tau}
	\tfrac{1}{\sqrt{\blx}}\left(\tfrac{\delta}{\sqrt{2\pi \amn{2}}}+2\tfrac{0.56 \amn{3}}{\amn{2}\sqrt{\amn{2}}}\right)
	\tfrac{1}{1-e^{-\rno\delta}}.
	\end{align}
Then for any \(\delta>0\), 
\eqref{eq:lem:SC-HTBE-achievability-q} holds for 
small enough \(\tau\). 
We choose the value of \(\delta\) considering the
constraint given in \eqref{eq:lem:SC-HTBE-achievability-w}.
Since \(\cln{\tin}\)'s are jointly independent under the probability measure \(\wma{\rno}{\mQ}\)
	we can bound \(\wma{\rno}{\mQ}(\set{B}_{\knd})\)
	from below using the Berry-Esseen theorem, \cite{berry41,esseen42,shevtsova10}. If \(\knd\) 
	satisfies 
	both
	\(-\sqrt{\amn{2}\blx}\leq \tau+\knd \delta\)
	and
	\(\tau+(\knd+1) \delta\leq \sqrt{\amn{2}\blx}\) 
	then
	\begin{align}
	\notag
	\wma{\rno}{\mQ}(\set{B}_{\knd})
	&\geq 
	\GCD{\tfrac{\tau+(\knd+1)\delta}{\sqrt{\amn{2}\blx}}}-\GCD{\tfrac{\tau+\knd \delta}{\sqrt{\amn{2}\blx}}}
	-2\tfrac{0.56}{\sqrt{\blx}} 
	\tfrac{\amn{3}}{\amn{2}\sqrt{\amn{2}}}
	\\
	\notag
	&=\tfrac{1}{\sqrt{2\pi}} \int_{\tfrac{\tau+\knd\delta}{\sqrt{\amn{2}\blx}}}^{\tfrac{\tau+(\knd+1)\delta}{\sqrt{\amn{2}\blx}}}
	e^{-\sfrac{\dsta^2}{2}} \dif{\dsta} -2\tfrac{0.56}{\sqrt{\blx}} \tfrac{\amn{3}}{\amn{2}\sqrt{\amn{2}}}
	\\
	\notag
	&\geq \tfrac{e^{-\frac{1}{2}}}{\sqrt{2\pi}} 
	\tfrac{\delta}{\sqrt{\amn{2}\blx}}
	-2\tfrac{0.56}{\sqrt{\blx}} \tfrac{\amn{3}}{\amn{2}\sqrt{\amn{2}}}
	\\
	\notag
	&=\tfrac{1}{\sqrt{\amn{2}\blx}}\left(\tfrac{\delta}{\sqrt{2\pi e}} -2 \tfrac{0.56\amn{3}}{\amn{2}}\right).
	\end{align}
	If \(\delta\!=\!e\sqrt{2\pi e\amn{2}}\htdelta\),
	then 
	for all \(\knd\) satisfying both
	\(\knd\geq \tfrac{-\sqrt{\amn{2}\blx}-\tau}{\delta} \)
	and
	\(\knd \leq \tfrac{\sqrt{\amn{2}\blx}-\tau-\delta}{\delta}\), 
	we have
	\begin{align}
	\label{eq:SC-HTBE-AB6}
	\wma{\rno}{\mQ}(\set{B}_{\knd})
	&\geq \tfrac{1}{\sqrt{2\pi\amn{2}\blx}}.
	\end{align}
Furthermore, \(\delta\geq \sqrt{e}\) for \(\delta\!=\!e\sqrt{2\pi e\amn{2}}\htdelta\),
and \eqref{eq:SC-HTBE-AB5} imply
	\begin{align}
	\notag
	\mQ(\oev)
&\leq e^{-\RD{1}{\wma{\rno}{\mQ}}{\mQ}-\rno \tau}
\tfrac{\htdelta}{\sqrt{\blx}}\left(e\sqrt{e}+e\right)
\tfrac{1}{1-e^{-\sqrt{e}}}
\\
\notag
&\leq e^{-\RD{1}{\wma{\rno}{\mQ}}{\mQ}-\rno \tau}
\tfrac{\htdelta}{\sqrt{\blx}}9
\end{align}
Then \(\oev\) satisfies \eqref{eq:lem:SC-HTBE-achievability-q} for
\(\tau=\tfrac{1}{\rno}\ln\tfrac{9\htdelta}{\rnb\sqrt{\blx}}\).
Furthermore,
the hypothesis of \eqref{eq:SC-HTBE-AB6} is satisfied 
for \(\knd=0\) when \(\tau=\tfrac{1}{\rno}\ln\tfrac{9\htdelta}{\rnb\sqrt{\blx}}\),
because \(\rnb\!\in\!\left[\tfrac{9\htdelta e^{\rno \delta}}{\sqrt{\blx}}
e^{-\rno\sqrt{\amn{2}\blx}},
\tfrac{9\htdelta}{\sqrt{\blx}} e^{\rno\sqrt{\amn{2}\blx}}\right]\),
by the hypothesis of the lemma. 
Thus \eqref{eq:SC-HTBE-AB4} and \eqref{eq:SC-HTBE-AB6} imply
	\begin{align}
	\notag
	\mW(\set{B}_{0})
	&\geq \tfrac{1}{\sqrt{2\pi\amn{2}\blx}}
	\left(\tfrac{9 \htdelta}{\rnb \sqrt{\blx}}\right)^{\frac{1-\rno}{\rno}}
	e^{-\RD{1}{\wma{\rno}{\mQ}}{\mW}+(1-\rno)e\sqrt{2\pi e\amn{2}}\htdelta}
	.
	\end{align}
	Then \(\oev\) satisfies \eqref{eq:lem:SC-HTBE-achievability-w}
	as a result of \eqref{eq:SC-HTBE-BZ}  and \eqref{eq:SC-HTBE-AB2}.
\end{proof}

\IEEEtriggeratref{26}
\bibliographystyle{unsrt}
\newcommand{\noopsort}[1]{} \newcommand{\printfirst}[2]{#1}
  \newcommand{\singleletter}[1]{#1} \newcommand{\switchargs}[2]{#2#1}

\end{document}